\documentclass{IEEEtran4PSCC}
\ifCLASSINFOpdf
   \usepackage[pdftex]{graphicx}
\else
   \usepackage[dvips]{graphicx}
\fi
\usepackage{cuted} 
\usepackage[cmex10]{amsmath}
\usepackage{amsfonts}
\usepackage{amssymb}
\usepackage{amsthm}
\newtheorem{theorem}{Theorem}
\newtheorem{lemma}{Lemma}

\usepackage{accents}
\newcommand{\ubar}[1]{\underaccent{\bar}{#1}}

\usepackage{float}
\ifCLASSOPTIONcompsoc
 \usepackage[caption=false,font=normalsize,labelfont=sf,textfont=sf]{subfig}
\else
 \usepackage[caption=false,font=footnotesize]{subfig}
\fi

\floatstyle{ruled}
\newfloat{model}{thp}{lop}
\floatname{model}{Model}

\usepackage{booktabs}

\usepackage{todonotes}

\DeclareMathOperator{\conv}{conv}

\usepackage{enumitem}
\newlist{abbrv}{itemize}{1}
\setlist[abbrv,1]{label=,labelwidth=0.5in,align=parleft,itemsep=0.1\baselineskip,leftmargin=!}

\begin{document}
%
\title{Strong Mixed-Integer Formulations for Transmission Expansion Planning with FACTS Devices}

\author{
\IEEEauthorblockN{Kevin Wu, Mathieu Tanneau, Pascal Van Hentenryck}
\IEEEauthorblockA{
    Georgia Institute of Technology, Atlanta, GA, United States \\
    kwu381@gatech.edu, \{mathieu.tanneau, pascal.vanhentenryck\}@isye.gatech.edu}
}


\maketitle

\begin{abstract}
    Transmission Network Expansion Planning (TNEP) problems find the most economical way of expanding a given grid given long-term growth in generation capacity and demand patterns.
    The recent development of Flexible AC Transmission System (FACTS) devices, which can dynamically re-route power flows by adjusting individual branches' impedance, call for their integration into TNEP problems.
    However, the resulting TNEP+FACTS formulations are significantly harder to solve than traditional TNEP instances, due to the nonlinearity of FACTS behavior.
    This paper proposes a new mixed-integer formulation for TNEP+FACTS, which directly represents the change in power flow induced by individual FACTS devices.
    The proposed formulation uses an extended formulation and facet-defining constraints, which are stronger than big-M constraints typically used in the literature.
    The paper conducts numerical experiments on a synthetic model of the Texas system with high renewable penetration.
    The results demonstrate the computational superiority of the proposed approach, which achieves a 4x speedup over state-of-the-art formulations, and highlight the potential of FACTS devices to mitigate congestion.
\end{abstract}


\thanksto{\noindent This research is partly funded by NSF award 2112533 and 2330450 and ARPA-E PERFORM award AR0001136.}

\section*{Nomenclature}
\label{sec:nomenclature}

\newcommand{\db}{\delta B}
\newcommand{\dbhi}{\bar{\delta B}}
\newcommand{\dblo}{\ubar{\delta B}}
\newcommand{\dva}{{\Delta \theta}}
\newcommand{\dvahi}{\bar{\Delta \theta}}
\newcommand{\dvalo}{\ubar{\Delta \theta}}

\paragraph{Sets}

\renewcommand{\S}{\mathcal{S}}
\newcommand{\N}{\mathcal{N}}
\newcommand{\E}{\mathcal{E}}
\newcommand{\T}{\mathcal{T}}

\begin{abbrv}
    \item[$\N$] Set of buses; $\N = \{1, ..., N\}$
    \item[$\E$] Set of branches; $\E = \{1, ..., E\}$
    \item[$\S$] Set of scenarios; $\S = \{1, ..., S\}$
\end{abbrv}

\paragraph{Parameters}

\newcommand{\pd}{\mathbf{p}^{\text{d}}}
\newcommand{\pgmin}{\mathbf{\ubar{p}}^{\text{g}}}
\newcommand{\pgmax}{\mathbf{\bar{p}}^{\text{g}}}
\newcommand{\pfmax}{\mathbf{\bar{p}}^{\text{f}}}
\newcommand{\dvamin}{\ubar{\theta}}
\newcommand{\dvamax}{\bar{\theta}}

\newcommand{\df}{\delta \pf}
\newcommand{\dX}{\delta X}
\newcommand{\dXmin}{\ubar{\delta X}}
\newcommand{\dXmax}{\bar{\delta X}}
\newcommand{\dB}{\delta B}
\newcommand{\dBmin}{\ubar{\delta B}}
\newcommand{\dBmax}{\bar{\delta B}}

\begin{abbrv}
    \item[$\pd_{i,s}$] Load at node $i \in \N$ in scenario $s \in \S$
    \item[$\pgmin_{i,s}, \pgmax_{i,s}$] min/max output of generator $i \, {\in} \, \N$ in scenario $s \, {\in} \, \S$
    \item[$c^{\text{cap}}_{ij}$] Line capacity upgrade cost of edge $ij \in \E$
    \item[$c^{\text{TCSC}}_{ij}$] TCSC installation cost on edge $ij \in \E$
    \item[$c_{i}$] Production cost function of generator $i$
    \item[$\lambda$] Power balance violation penalty cost
    \item[$\Delta^c$] Capacity upgrade increment (MW) of branches
    \item[$m$] Number of capacity upgrade increments available
    \item[$X_{ij}$] Reactance of branch $ij \in \E$
    \item[$\pfmax_{ij}$] Thermal limit of branch $ij \in \E$
    \item[$\dvamin_{ij}, \dvamax_{ij}$] min/max angle difference on branch $ij \in \E$
\end{abbrv}

\paragraph{Variables}

\newcommand{\pg}{\mathbf{p}^{\text{g}}}
\newcommand{\pf}{\mathbf{p}^{\text{f}}}
\newcommand{\va}{\theta}

\begin{abbrv}
    \item[$\gamma_{ij}$] Capacity upgrade level of branch $ij \in \E$
    \item[$\psi_{ij}$] Whether a FACTS device is installed on branch $(i,j) \in \E$
    \item[$\va_{i, s}$] Voltage angle of bus $i \in \N$ in scenario $s \in \S$
    \item[$\pg_{i,s}$] Output of generator $i \in \N$ in scenario $s \in \S$
    \item[$\pf_{ij,s}$] Power flow on branch $ij \in \E$ in scenario $s \in \S$
    \item[$\xi_{i,s}$] Power imbalance at bus $i\in \N$ in scenario $s \in \S$
\end{abbrv}

\section{Introduction}
\label{sec:introduction}

In order to accommodate a high penetration of renewable generation,
especially wind and solar generation, existing power grids need
substantial investments in transmission capacity \cite{RTE}.  This
poses a challenge for transmission system operators (TSOs) who are
responsible for the expansion of infrastructure that delivers power
from generators to loads.  This need naturally gives rise to
Transmission Network Expansion Planning (TNEP) problems, which find
the most economical way of expanding the grid given long-term growth
in generation capacity and demand patterns
\cite{Gideon2019_SurveyTNEP}. TNEP problems are mixed-integer
programming problems, which are notoriously hard to solve
\cite{Gideon2019_SurveyTNEP}.

Recent years have also witnessed the development of so-called Flexible AC
Transmission System (FACTS) devices, such as phase-shifting
transformers (PSTs), Thrystor Controlled Series Compensators (TCSCs),
and Static Synchronous Series Compensators (SSSCs), to name a few
\cite{Mokhtari}.  FACTS devices can dynamically re-route power flows
across the grid by emulating a change in line impedances.  It is
therefore natural to include such technology in TNEP studies, which
can lead to lower investment and operational costs, and improve the
long-term reliability of the grid.  However, TNEP problems with FACTS
devices (TNEP+FACTS for short) have received little attention in the
literature.  {\em This paper addresses this gap by proposing new, improved
mixed-integer linear programming (MILP) formulations for the
TNEP+FACTS problem.}

\subsection{Related Work}
\label{sec:introduction:literature}

Because of its importance for long-term power systems operations,
there is a rich body of works on TNEP problems and their solution,
e.g., \cite{Li,Micheli,Lumbreras,Blanchot}.  TNEP problems may
consider transmission expansion only, or combined generation and
transmission expansion planning \cite{Li}.  Power flow equations are
typically linearized using the DC approximation, which yields MILP
problems.  Stochastic formulations of TNEP, e.g. two-stage stochastic
programming formulations \cite{Micheli,Lumbreras}, are also popular.
Most state-of-the-art approaches use a form of Benders decomposition
algorithm to solve TNEP efficiently.  Readers are referred to
\cite{Gideon2019_SurveyTNEP} for a more exhaustive survey of TNEP
problems.  The paper focuses on TNEP+FACTS problems, which have
received significantly less attention.

Several papers combine traditional TNEP methodologies with the
deployment of FACTS devices
\cite{Mokhtari,de_Araujo,Luburic,Esmaili,Ziaee,Franken}.  Mokhtari et
al. \cite{Mokhtari} consider a distributionally robust formulation
with linearized AC constraints, which is solved with a Benders'
decomposition algorithm.  A differential evolution-based metaheuristic
algorithm is proposed in \cite{de_Araujo} for solving TNEP+FACTS with
AC constraints.  Also in the AC setting, Luburic et al. \cite{Luburic} consider a
TNEP+FACTS formulation with TCSCs and energy storage, and report
results on a 24-bus system.  Esmaili et al. \cite{Esmaili} consider a multi-stage,
short circuit-constrained TNEP problem with TCSC and SFCL devices,
which is executed on 39-bus and 118-bus systems.  It should be noted
that the studies mentioned above were conducted solely on small
artificial test systems which may not capture the dynamics of
real-life grids.

Ziaee et al. \cite{Ziaee} consider the placement of TCSC devices
within a TNEP framework in the DC setting.  The TNEP+FACTS problem is
modeled via a combination of big-M and third-order relaxation
techniques, and experiments are reported on 6-bus and 118-bus
systems.  Franken et al. \cite{Franken} propose an MILP formulation for
TNEP+FACTS where the effect of FACTS devices is modeled as virtual
phase shifts.  The resulting disjunctive constraints are linearized
using big-M constraints.  Numerical results on a synthetic 120-bus
system suggest that FACTS devices have the potential to reduce
expansion costs, renewable curtailment, and load shedding.  One
limitation of the study in \cite{Franken} is that the proposed
TNEP+FACTS model considers only a single load/renewable scenario, which may
fail to capture, e.g., seasonal variations in system conditions.

\subsection{Contributions}
\label{sec:introduction:contribution}

The paper considers the integration of FACTS devices in the
formulation and solution of TNEP problems, and makes the following
contributions.  First, it proposes to formulate for TNEP+FACTS using
variables that directly represents the impact of FACTS devices on
power flows, rather than using virtual angle shifts.  Second, it
introduces a strong MILP formulation for TNEP+FACTS using an extended
formulation with facet-defining inequalities.  Third, it conducts
numerical experiments on the Texas system under a high penetration of
renewable generation. {\em The results demonstrate that the proposed
  MILP formulation achieves significant speedups over current
  state-of-the-art formulations, with a 4x speedup and a 40x reduction
  in branch-and-bound nodes needed to prove optimality.}  The
experiments also provide some insights into the benefits of FACTS
devices for long-term network expansion planning. In particular, {\em
  the results highlight the benefits of FACTS devices for expansion
  planning through flow control.} Together, these results should encourage
subsequent studies to include FACTS devices in their studies and adopt
the strong formulation presented in the paper.

The rest of the paper is organized as follows.  Section
\ref{sec:formulation} describes the TNEP and TNEP+FACTS formulations
considered in the paper, presents the proposed strong MILP formulation
for TNEP+FACTS, and proves its theoretical properties.  Section
\ref{sec:results} presents a case study of TNEP+FACTS on a synthetic model of the Texas
system, which demonstrates the computational advantages of the proposed
formulation, and illustrates the benefits of FACTS devices on a real
system.  Section \ref{sec:conclusion} concludes the paper and
highlights directions for future work.

\section{Formulation}
\label{sec:formulation}

This section presents the TNEP formulations used in the paper.
Section \ref{sec:formulation:TNEP} presents the baseline TNEP formulation, which does not include FACTS devices.
Section \ref{sec:formulation:TNEP:FACTS} presents an initial MILP formulation for TNEP+FACTS, based on \cite{Franken}, which includes FACTS devices like TCSCs.
Section \ref{sec:formulation:TNEP:FACTS:strong} presents the proposed strengthened MILP formulation for TNEP+FACTS, and proves theoretical guarantees for the quality of the formulation.

The presentation assumes, for ease of reading and without loss of generality, that exactly one generator is attached to each bus.
Reference (slack) bus voltage constraints are also omitted for simplicity.

\subsection{Baseline TNEP formulation}
\label{sec:formulation:TNEP}

    Model \ref{model:TNEP} presents the baseline TNEP formulation, whose variables, constraints and objective are described next.
    
    \begin{model}[!t]
        \caption{The TNEP formulation without FACTS devices}
        \label{model:TNEP}
        \begin{align*}
            \text{min} \quad 
                & \sum_{ij \in \E} c^{\text{cap}}_{ij} \gamma_{ij}
                + \frac{1}{S} \sum_{i \in \N, s \in \S} 
                    c_{i}(\pg_{i,s})
                    + \lambda |\xi_{i,s}|
                \\
            \text{s.t.} \quad 
            & \eqref{eq:TNEP:power_balance}-\eqref{eq:TNEP:generation:min_max_limits},
                \quad \forall i \in \N, s \in \S \\
            & \eqref{eq:TNEP:ohm}-\eqref{eq:TNEP:phase_angle_difference},
                \quad \forall ij \in \E, s \in \S
        \end{align*}
    \end{model}

    \subsubsection{Variables}
    \label{sec:formulation:TNEP:variables}
        
        An important feature of the paper's approach is the observation that, due to a combination of economic, political and social factors, new transmission lines are very uncommon.
        Therefore, \emph{the paper only considers capacity upgrades for existing lines}, which are captured by investment variables $\gamma_{ij}$ that denote, for each branch $ij \in \E$ the capacity upgrade level on that branch.
        Investment variables $\gamma_{ij}$ are restricted to integer values, with up to $m$ possible capacity upgrades for each line.
        Other variables in the problem include generation dispatches $\pg$, nodal voltage angles $\va$, power flows $\pf$, and nodal imbalance variables $\xi$.

    \subsubsection{Constraints}
    \label{sec:formulation:TNEP:constraints}
        Power balance is enforced at each bus $i \, {\in} \, \N$ and scenario $s \, {\in} \, \S$ as follows:
        \begin{align}
            \label{eq:TNEP:power_balance}
            \pg_{i,s} + \sum_{ji \in \E} \pf_{ji,s} - \sum_{ij \in \E} \pf_{ij,s} = \pd_{is} + \xi_{i,s},
        \end{align}
        where slack variable $\xi$ captures nodal power imbalances.
        
        The generation minimum and maximum limits are enforced for each generator $i \in \N$ and scenario $s \in \S$ as
        \begin{align}
            \label{eq:TNEP:generation:min_max_limits}
            \pgmin_{i,s} \leq \pg_{i,s} \leq \pgmax_{i,s}.
        \end{align}
        Renewable (wind and solar) generators have a minimum output of
        0, and their maximum output varies based on the scenario $s
        \in \S$.  The minimum/maximum limits of non-renewable
        generators typically depend on the physical
        characteristics of the generator, and may vary between, e.g., summer and
        winter.
        
        The power flow on branch $ij \, {\in} \, \E$ in scenario $s \, {\in} \, \S$ is determined by Ohm's law
        \begin{align}
            \label{eq:TNEP:ohm}
            \pf_{ij,s} = \frac{1}{X_{ij}} (\theta_{j, s} - \theta_{i, s}),
        \end{align}
        and must satisfy thermal limit constraints of the form
        \begin{align}
            \label{eq:TNEP:thermal_limit}
            -\pfmax_{ij} - \gamma_{ij} \Delta^{c}_{ij} 
            \leq \pf_{ij,s}
            \leq \pfmax_{ij} + \gamma_{ij} \Delta^{c}_{ij}.
        \end{align}
        The original thermal limit $\pfmax_{ij}$ may be
        increased by upgrading the line, as captured by variable
        $\gamma_{ij}$.  In addition, phase angle differences
        $\va_{ij,s} \, {=} \, (\va_{j,s} \, {-} \, \va_{i,s})$ are
        constrained by
        \begin{align}
            \label{eq:TNEP:phase_angle_difference}
            \dvamin_{ij} \leq \va_{ij,s} \leq \dvamax_{ij}.
        \end{align}

    \subsubsection{Objective}
    \label{sec:formulation:TNEP:objective}
    
        The TNEP objective minimizes up-front investment costs in capacity upgrades, plus operational costs that capture, for each scenario, generation costs and penalty costs for any unserved/overserved energy.

\subsection{TNEP with FACTS devices}
\label{sec:formulation:TNEP:FACTS}

The paper also considers FACTS devices such as TCSCs, which can
dynamically adjust the reactance of a line \cite{Mokhtari}.  Binary
$\psi_{ij} \, {\in} \, \{0, 1\}$ takes value 1 if a FACTS device is
installed on branch $ij$, and 0 otherwise.  If a FACTS device is
installed on branch $ij \, {\in} \, \E$, i.e., if $\psi_{ij} \, {=} \,
1$, Ohm's law \eqref{eq:TNEP:ohm} in scenario $s \, {\in} \, \S$
becomes
\begin{align}
    \label{eq:TNEP:FACTS:ohm}
    \pf_{ij,s} 
        &= \frac{1}{X_{ij} + \delta X_{ij, s}} \va_{ij,s},\\
    \label{eq:TNEP:FACTS:ohm:susceptance}
        &= (B_{ij} + \dB_{ij, s}) \va_{ij, s},
\end{align}
where variables $\dX_{ij, s}$ and $\dB_{ij,s}$ capture the change in reactance and susceptance induced by the FACTS device on branch $ij$ in scenario $s$, where
\begin{align}
    \dB_{ij, s} = \frac{- \dX_{ij}}{X_{ij} (X_{ij} + \dX_{ij})}.
\end{align}
Equation \eqref{eq:TNEP:FACTS:ohm:susceptance} can be written equivalently as
\begin{align}
    \label{eq:TNEP:FACTS:flow_representation}
    \pf_{ij,s} &= B_{ij} \va_{ij,s} + \df_{ij,s},
\end{align}
where variable $\df_{ij,s} \, {=} \, \dB_{ij,s} \theta_{ij,s}$ captures the change in power flow induced by the FACTS device.

\begin{figure}[!t]
    \centering
    \begin{tikzpicture}[xscale=1.3,yscale=1.3]
        \fill[fill=blue!30, opacity=0.5] (0,0) -- (2,1) -- (2,-0.5) -- cycle;
    
        \fill[fill=red!30, opacity=0.5] (0,0) -- (-2,0.5) -- (-2,-1) -- cycle;
    
        \draw[->,thick] (-2.5,0) -- (2.5,0) node[right] {$\theta$};
        \draw[->,thick] (0,-1.5) -- (0,1.5) node[above] {$\df$};

        \node[below, xshift=5pt] at (2,0) { $\dvamin$ };
        \node[below, xshift=-5pt] at (-2,0) { $\dvamax$ };

        \draw[opacity=0] (-2, 0.5) -- (2, -0.5) node[pos=0.25, above, sloped, opacity=1] { $\dBmin$ };
        \draw[opacity=0] (-2, -1) -- (2, 1) node[pos=0.75, above, sloped, opacity=1] { $\dBmax$ };

        
    \end{tikzpicture}
    \caption{Domain Space of $(\pf, \theta)$}
    \label{fig:pf_theta_domain}
\end{figure}
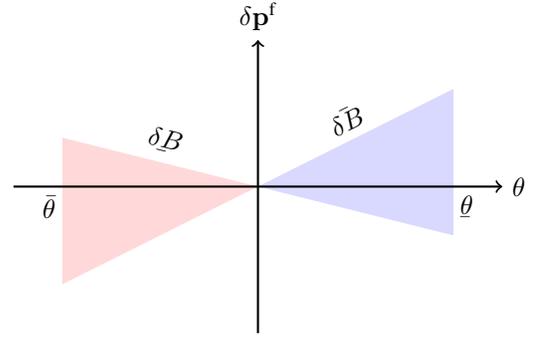

Substituting out $\dB_{ij,s} \,{=}\, (\df_{ij,s}) \va_{ij,s}^{-1}$ yields the disjunction
\begin{subequations}
    \label{eq:TNEP:FACTS:disjunction}
    \begin{align}
        \dBmin_{ij} \va_{ij,s} \leq \df_{ij,s} & \leq \dBmax_{ij} \va_{ij,s} && \text{ if } \theta_{ij,s} \geq 0,\\
        \dBmax_{ij} \va_{ij,s} \leq \df_{ij,s} & \leq \dBmin_{ij} \va_{ij,s} && \text{ if } \theta_{ij,s} \leq 0,
    \end{align}
\end{subequations}
where $\dBmin_{ij}, \dBmax_{ij}$ are the minimum and maximum changes in susceptance induced by the FACTS device.
The disjunction is illustrated in Figure \ref{fig:pf_theta_domain}.
Franken et al. \cite{Franken} represent the disjunction \eqref{eq:TNEP:FACTS:disjunction} via big-M constraints, and introduce binary variable $z_{ij,s}$ that takes value 1 if $\theta_{ij,s} \geq 0$ and 0 otherwise.
This formulation reads
\begin{align}
    \label{eq:TNEP:FACTS:flowdiff_bigM_positive}
        \df_{ij,s} &\geq \dBmin_{ij} \va_{ij,s} - M (1 - z_{ij,s}),\\
    \df_{ij,s} &\leq \dBmax_{ij} \va_{ij,s} + M (1 - z_{ij,s}),\\
    \df_{ij,s} &\geq \dBmax_{ij} \va_{ij,s} - M z_{ij,s},\\
    \label{eq:TNEP:FACTS:flowdiff_bigM_negative}
        \df_{ij,s} &\leq \dBmin_{ij} \va_{ij,s} + M z_{ij,s}.
\end{align}
In addition, to ensure that power flows can only be modified when a FACT device is installed, \cite{Franken} use the additional big-M constraint
\begin{align}
    \label{eq:TNEP:FACTS:installation_bigM}
    -M \psi_{ij} \leq \df_{ij,s} \leq M \psi_{ij}. 
\end{align}
The resulting TNEP+FACTS formulation is presented in Model \ref{model:TNEP+FACT:franken}.
Note that the formulation stated in \cite{Franken} does not use the $\df$ variables, but an equivalent angle representation.

\begin{model}[!t]
    \caption{The baseline TNEP+FACTS formulation}
    \label{model:TNEP+FACT:franken}
    \small
    \begin{align*}
    \min \quad 
        & 
            \sum_{ij \in \E} c^{\text{cap}}_{ij} \gamma_{ij} + c^{\text{TCSC}}_{ij} \psi_{ij} 
            + \frac{1}{S} \sum_{i \in \N, s \in \S} c_{i}(\pg_{i,s}) + \lambda |\xi_{i,s}|\\
        & \nonumber \\
        \text{s.t.} \quad 
        & \eqref{eq:TNEP:power_balance}, \eqref{eq:TNEP:generation:min_max_limits} 
            \quad \forall i \in \N, s \in \S \\
        & \eqref{eq:TNEP:thermal_limit}, \eqref{eq:TNEP:phase_angle_difference}, \eqref{eq:TNEP:FACTS:flow_representation}
            \quad \forall ij \in \E, s \in \S \\
        & \eqref{eq:TNEP:FACTS:flowdiff_bigM_positive}-\eqref{eq:TNEP:FACTS:flowdiff_bigM_negative} 
            \quad \forall ij \in \E, s \in \S \\
        & \eqref{eq:TNEP:FACTS:installation_bigM} 
            \quad \forall ij \in \E, s \in \S
    \end{align*}
\end{model}

The original presentation in \cite{Franken} uses a single constant $M$ for big-M constraints \eqref{eq:TNEP:FACTS:flowdiff_bigM_positive}--\eqref{eq:TNEP:FACTS:installation_bigM}, whose value is not specified.
A valid value for $M$ is $\max_{ab \in \E}(\pfmax_{ab} + m\Delta^c)$.
Tight big-M bounds are important in ensuring strong performance of
MILP solvers.  Therefore, this paper proposes to improve the approach
of \cite{Franken} by using a different big-M bound for each branch
$ M_{ij} = \pfmax_{ij} + m\Delta^c$.
This modified formulation uses the same variables and constraints as
the formulation outlined in Model \ref{model:TNEP+FACT:franken}; only
the big-M bounds are modified.

\subsection{Strong MILP Formulation for TNEP+FACTS}
\label{sec:formulation:TNEP:FACTS:strong}

While classical TNEP problems can be solved using Benders
decomposition, this strategy is not applicable when considering FACTS
devices because of the non-convexity of the disjunctions in
\eqref{eq:TNEP:FACTS:disjunction}.  Instead, the TNEP+FACTS problem
is solved directly using off-the-shelf MILP solvers.  In that context,
the formulation proposed by \cite{Franken} has two main limitations.
On the one hand, it relies on big-M constraints, which are notoriously
detrimental to the performance of MILP solvers, because they typically
lead to weak linear relaxations with highly fractional solutions.  On
the other hand, there is no direct link between binary variables
$\psi$ and $z$, with the latter introducing an un-necessary level of
complexity when $\psi = 0$.

\emph{A key contribution of this paper is to address this limitation
  by proposing a strong MILP formulation for the TNEP+FACTS problem},
which significantly improves solution time.  The proposed formulation
uses an extended formulation of the disjunctive constraints
\eqref{eq:TNEP:FACTS:disjunction} using an additional binary variable,
together with facet-defining inequalities.  The strong formulation is
presented for a fixed branch $ij \in \E$ and scenario $s \in \S$.  For
each of reading, the $ij$ and $s$ subscripts are dropped; note that
$\va$ thus denotes the angle difference $\va_{ij,s} \, {=} \,
\va_{j,s} {-} \va_{i,s}$.

    First observe that, to capture the fact that $\df = 0$ if no FACTS device is installed, the 2-term disjunction in \eqref{eq:TNEP:FACTS:disjunction} can be extended to the 3-term disjunction
    \begin{subequations}
        \label{eq:TNEP:FACTS:disjunction:strong}
        \begin{align}
            \df & = 0 && \text{ if } \psi = 0,\\
            \dBmin \va \leq \df & \leq \dBmax \va && \text{ if } \psi = 1 \text{ and } \theta \geq 0,\\
            \dBmax \va \leq \df & \leq \dBmin \va && \text{ if } \psi = 1 \text{ and } \theta \leq 0.
        \end{align}
    \end{subequations}
    This 3-term disjunction can be represented as the union of the following three polyhedra
    \begin{align}
        \mathcal{P}_{0,0,0} :&
        \left\{
        \begin{array}{l}
            (\psi, z^{+}, z^{-}) = (0, 0, 0)\\
             \df = 0\\
             \dvamin \leq \va \leq \dvamax
        \end{array}
        \right. \\
        \mathcal{P}_{1, 1, 0} :&
        \left\{
        \begin{array}{l}
            (\psi, z^{+}, z^{-}) = (1, 1, 0)\\
            \dBmin \cdot\va \leq \df \leq \dBmax \cdot\va\\
            0 \leq \va \leq \dvamax
        \end{array}
        \right.\\
        \mathcal{P}_{1, 0, 1} :&
        \left\{
        \begin{array}{l}
            (\psi, z^{+}, z^{-}) = (1, 0, 1)\\
            \dBmax \cdot \va \leq \df \leq \dBmin \cdot\va\\
            \dvamin \leq \va \leq 0
        \end{array}
        \right.
    \end{align}
    where $z^{+}, z^{-}$ are additional binary variables that capture the relationship between $\psi, \theta$ and $\df$.
    Namely, if $\psi \, {=} \, 0$, then $z^{+}, z^{-}$ are both set to zero.
    Conversely, when $\psi \, {=} \, 1$, variable $z^{+}$ (resp. $z^{-}$) indicates whether the phase angle difference is positive (resp. negative), i.e., $\theta \geq 0$, (resp. $\leq 0$).

The MILP formulation for the 3-term disjunction proposed in this paper 
\eqref{eq:TNEP:FACTS:disjunction:strong} is given by
    \begin{subequations}
    \label{eq:TNEP:FACTS:facets}
    \begin{align}
        z^{+} + z^{-} &= \psi,\\
        \dvamin z^{+} + \va & \geq \dvamin, \\
        \dvamax z^{-} + \va & \leq \dvamax, \\
        - \dvamax \, \dBmin z^{+} - \dvamin \, \dBmax z^{-} + \df & \geq 0,\\
        - \dvamax \, \dBmax z^{+} - \dvamin \, \dBmin z^{-} + \df & \leq 0, \\
        \dvamax \, \dBmin z^{+} + \dvamax \, \dBmax z^{-} + \dBmax \va - \df & \leq \dvamax \, \dBmax, \\
        \dvamax \, \dBmax z^{+} + \dvamax \, \dBmin z^{-} + \dBmin \va - \df & \geq \dvamax \, \dBmin, \\
        \dvamin \, \dBmax z^{+} + \dvamin \, \dBmin z^{-} + \dBmax \va - \df & \geq \dvamin \, \dBmax, \\
        \dvamin \, \dBmin z^{+} + \dvamin \, \dBmax z^{-} + \dBmin \va - \df & \leq \dvamin \, \dBmin, \\
        \psi, z^{+}, z^{-} & \in \{0, 1\},
    \end{align}
    \end{subequations}
    and it is complemented by the big-M constraint
    \begin{align}
        \label{eq:TNEP:FACTS:strong:bigM}
        - (\pfmax + m \Delta^{c}) \psi \leq \df \leq (\pfmax + m \Delta^{c}) \psi,
    \end{align}
which, although redundant with constraints
\eqref{eq:TNEP:FACTS:facets},
\eqref{eq:TNEP:FACTS:flow_representation} and
\eqref{eq:TNEP:thermal_limit}, was found to improve performance.  The
proposed strong TNEP+FACTS formulation is summarized in Model
\ref{model:TNEP+FACT:convex_hull}.
    \begin{model}[!t]
        \caption{The proposed TNEP+FACTS formulation}
        \label{model:TNEP+FACT:convex_hull}
        \small
        \begin{align*}
        \min \quad 
            & 
                \sum_{ij \in \E} c^{\text{cap}}_{ij} \gamma_{ij} + c^{\text{TCSC}}_{ij} \psi_{ij} 
                + \frac{1}{S} \sum_{i \in \N, s \in \S} c_{i}(\pg_{i,s}) + \lambda |\xi_{i,s}|\\
            & \nonumber \\
            \text{s.t.} \quad 
            & \eqref{eq:TNEP:power_balance}, \eqref{eq:TNEP:generation:min_max_limits} 
                \quad \forall i \in \N, s \in \S \\
            & \eqref{eq:TNEP:thermal_limit}, \eqref{eq:TNEP:phase_angle_difference}, \eqref{eq:TNEP:FACTS:flow_representation}
                \quad \forall ij \in \E, s \in \S \\
            & \eqref{eq:TNEP:FACTS:facets}
                \quad \forall ij \in \E, s \in \S\\
            & \eqref{eq:TNEP:FACTS:strong:bigM}
                \quad \forall ij \in \E, s \in \S
        \end{align*}
    \end{model}

    It is easy to verify, e.g., by explicit inspection of all three feasible realizations of $\psi, z^{+}, z^{-}$, that the feasible set of constraints \eqref{eq:TNEP:FACTS:facets} is exactly $\mathcal{P}_{0,0,0} \cup \mathcal{P}_{1,1,0} \cup \mathcal{P}_{1,0,1}$.
    This, in turn, proves the validity of the proposed formulation (Model \ref{model:TNEP+FACT:convex_hull}).
    The strength of the formulation stems from the result of Theorem \ref{thm:facets} below, which shows that every constraint in \eqref{eq:TNEP:FACTS:facets} is facet-defining, i.e., the constraints are as tight as possible.
    \begin{theorem}
        \label{thm:facets}
        The inequality constraints in \eqref{eq:TNEP:FACTS:facets} are facet-defining for $\mathcal{P} = \conv(\mathcal{P}_{0,0,0} \cup \mathcal{P}_{1,1,0} \cup \mathcal{P}_{1,0,1})$.
    \end{theorem}
    \begin{proof}
        First, observe that $\mathcal{P}$ is a 4-dimensional set
        because of the equality constraint $\psi = z^{+} + z^{-}$.
        Therefore, every facet of $\mathcal{P}$ is a valid inequality
        for $\mathcal{P}$ which is tight at 4 extreme points.  The
        rest of the proof follows by verifying that each inequality
        constraint in \eqref{eq:TNEP:FACTS:facets} is tight at least 4
        extreme points of $\mathcal{P}$, using the extreme point
        characterization of Lemma \ref{lemma:extreme_points}.  Note
        that $\psi \leq 1, z^{+} \geq 0$ and $z^{-} \geq 0$, which are
        a consequence of the binary requirement on $\psi, z^{+},
        z^{-}$, are also facet-defining.
    \end{proof}
    \begin{lemma}
        \label{lemma:extreme_points}
        The extreme points of $\mathcal{P}_{0, 0, 0}$, $\mathcal{P}_{1, 1, 0}$ and $\mathcal{P}_{1, 0, 1}$ are
        \begin{subequations}
        \begin{align}
            \mathcal{P}_{0,0,0} &= \conv \left\{
                \begin{array}{l}
                    (0, 0, 0, \dvamin, 0),\\
                    (0, 0, 0, \dvamax, 0)
                \end{array}
            \right\}\\
            \mathcal{P}_{1,1,0} &= \conv \left\{
                \begin{array}{l}
                    (1, 1, 0, 0, 0),\\
                    (1, 1, 0, \dvamax, \dvamax \cdot \dBmin),\\
                    (1, 1, 0, \dvamax, \dvamax \cdot \dBmax)
                \end{array}
            \right\}\\
            \mathcal{P}_{1,0,1} &= \conv \left\{
                \begin{array}{l}
                    (1, 0, 1, 0, 0),\\
                    (1, 0, 1, \dvamin, \dvamin \cdot \dBmin),\\
                    (1, 0, 1, \dvamin, \dvamin \cdot \dBmax)
                \end{array}
            \right\}
        \end{align}
        \end{subequations}
    \end{lemma}
    \begin{proof}
        Immediate by inspection.
    \end{proof}

\noindent
Finally, the paper implements a simple bound-tightening
procedure to further strengthen the formulation.  This strategy is
applied as a pre-processing step, and it tightens lower and upper
limits on each phase-angle difference.  Namely, the approach
uses the relation
    \begin{align}
        \pfmax + m \Delta^{c} \geq \pf &= (B + \dB) \va \geq (B + \dBmin) \va
    \end{align}
    which yields $\va \leq \min(\dvamax, (\pfmax + m \Delta^{c}) (B + \dBmin)^{-1})$ and allows the opportunity to tighten the upper limit on $\theta$.
    The lower limit is tightened similarly.

\section{Numerical results}
\label{sec:results}

\subsection{Data generation}
\label{sec:results:data}

The paper reports numerical experiments on the 123-bus Texas 345kv
backbone transmission system \cite{TX123BT}.  The system comprises 123
buses, 255 branches, 138 non-renewable generators, and 154 renewable
(wind and solar) generators.  The data provided in \cite{TX123BT} also
includes hourly time series of load and renewable generation.  To
reflect future systems with a high renewable penetration, the paper
scales the load, wind and solar time series by 1.5, 2 and 3,
respectively. This should give interesting case studies for how
the Texas grid will evolve. 
    
Figure \ref{fig:texas_renewables} depicts the grid topology, which
spans the state of Texas, and shows the location and size (scaled
nameplate capacity, in MW) of each renewable generator.  An important
feature of the Texas system is that renewable generation is primarily
located in the west of the state, where the wind resource
is abundant, whereas load areas are concentrated around major
metropolitan and industrial areas, especially Houston (south-east, by
the coast) and Dallas (north-east).

\begin{figure}[!t]
\centering
\includegraphics[width=0.80\columnwidth]{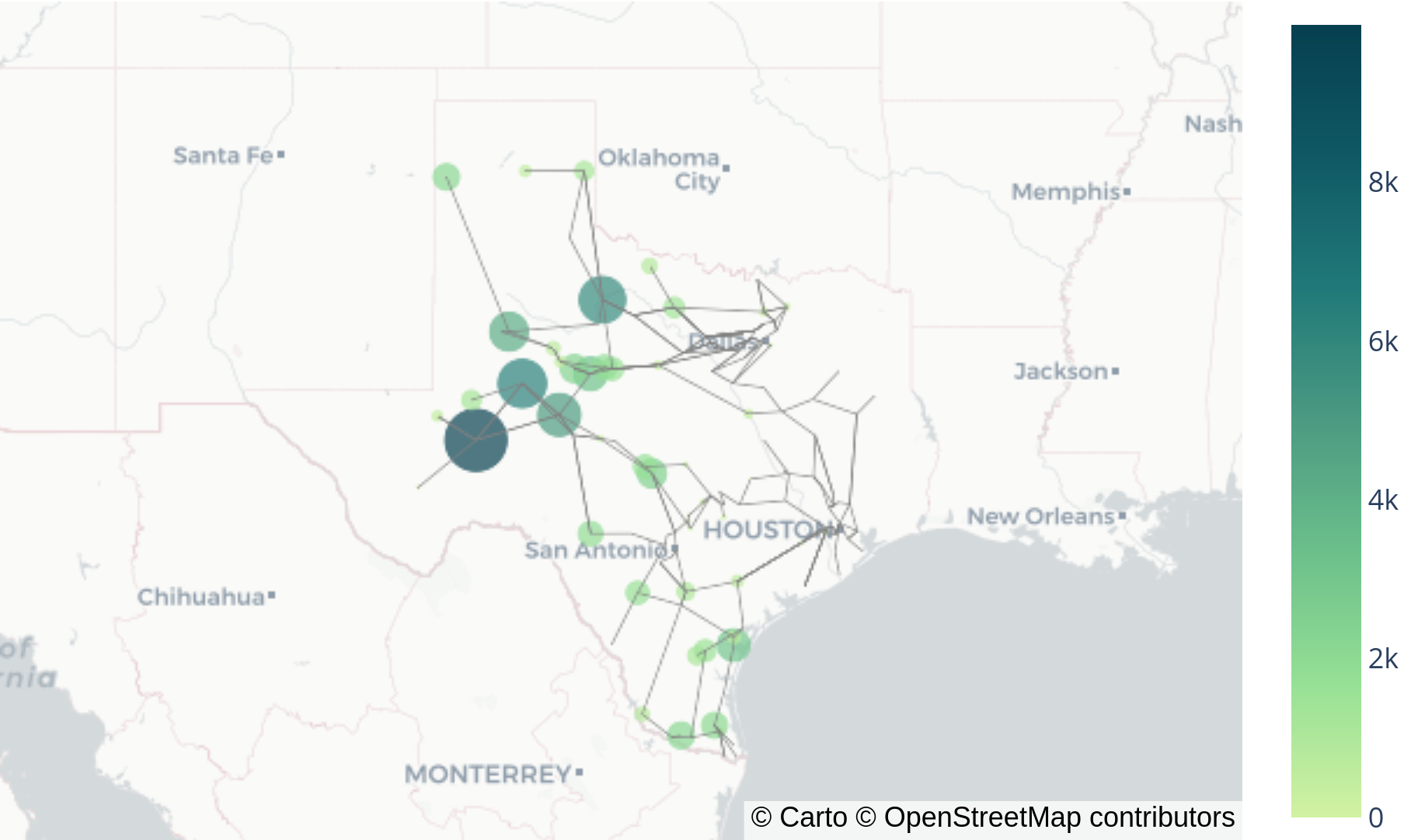}
\caption{The 123-bus Texas 345kV backbone topology. Circles indicate the location and size (scaled nameplate capacity, in MW) of renewable generators.}
\label{fig:texas_renewables}
\end{figure}

The experiments consist of 10 scenarios for TNEP, each corresponding
to one hour of the year.  For both summer and winter, the experiments
select the five hours corresponding to the highest load, highest net load,
highest wind generation, highest solar generation, and lowest wind
generation over that season.  These hours are the most likely to lead
to congestion, load shedding, and/or curtailment of renewable
generation.

The production costs of non-renewable generators are taken from
\cite{TX123BT}, and the cost of renewable generators is set to zero.
Table \ref{tab:cost_config} reports cost information regarding
capacity upgrades, TCSC devices, and unserved/overserved energy
penalties, which are adapted from \cite{Esmaili}.  There is little
information available in the literature regarding the cost of TCSC
devices.  Preliminary experiments using the same cost data as reported
in \cite{Esmaili}, resulted in no capacity upgrades nor FACTS device
being installed, i.e., it was more economical to pay penalties for
unserved energy.  This suggests that the penalties in \cite{Esmaili}
might have been too low relative to the upgrade costs.  Therefore, the
authors elected to adjust penalties and investment costs to obtain
more interesting test cases and more insights into potential upgrade
scenarios.

    \begin{table}[!t]
        \centering
        \caption{Cost Configuration}
        \label{tab:cost_config}
        \begin{tabular}{cr}
        \toprule
         & Costs \\
        \midrule
        Unserved \& Overserved Energy          & \$50,000/MWh       \\
            Capacity Upgrades &  \$124/MW-km       \\ 
            TCSCs           & \$2,200/MVA          \\ 
        \bottomrule
        \end{tabular}
    \end{table}

\subsection{Computational Performance}
\label{sec:results:performance}

All TNEP formulations are implemented in Julia using the JuMP
\cite{Lubin2023} modeling language. All experiments are executed on 24-core Intel Xeon
machines running Linux, on the Phoenix cluster \cite{PACE}.
They are executed with 24 threads, 288GB of RAM, a 6-hour time
limit, a 0.01\% optimality gap tolerance, and other settings set to
default. 

Four TNEP formulations are compared:
\begin{enumerate}
\item TNEP: the baseline TNEP
formulation of Model \ref{model:TNEP}, which does not consider FACTS
devices.
\item FBSM: the TNEP+FACTS formulation of Model
  \ref{model:TNEP+FACT:franken} \cite{Franken}.
\item FBSMi: the
improved FBSM formulation, with individual big-M values for
individual branches.
\item  FACeTS: the paper's formulation,
  presented in Model \ref{model:TNEP+FACT:convex_hull}.
\end{enumerate}

Experiments were conducted with two optimization solvers: Gurobi 10 \cite{gurobi}, a commercial solver, and HiGHS \cite{highs}, a leading open-source solver.
Despite the comprehensive capabilities of HiGHS, it was unable to find good-quality solutions within the allocated 6-hour time limit, for any of the four formulations.
Consequently, this study focuses on the results obtained through Gurobi, as it consistently provided results within the specified computational limits.

Table \ref{tab:computational_performance} reports, for each
formulation: the number of variables (Vars), constraints (Constrs),
and the corresponding compute time (Time, in seconds), and number of
branch-and-bound nodes.  The former two indicate the size of the
problem, and the latter two metrics indicate the computational
performance. 

    \begin{table}[t]
        \centering
        \caption{Comparison of MILP formulations}
        \label{tab:computational_performance}
        \begin{tabular}{crrrr}
            \toprule
             Formulation & Vars. & Constrs & Time  & Nodes \\
             \midrule
             TNEP & 9,415 & 13,980 & 4s  & 156 \\
             \midrule
             FBSM \cite{Franken} & 14,770 & 29,280 & 1,733s  & 268,940 \\
             FBSMi & 14,770 & 29,280 & 1,375s  & 253,954 \\
             FACeTS & 17,320 & 36,930 & 386s & 6,273 \\
             \bottomrule
        \end{tabular}
    \end{table}

As expected, including FACTS devices significantly increases computing
times, by a factor 100x--400x.  The FBSM formulation performs the
worst, with a computing time of 1,700s and over 260,000
branch-and-bound nodes needed to reach the prescribed optimality gap.
The use of tighter big-M bounds yields a reduction in computing time
of 20\%, and marginal reduction in branch-and-bound nodes. {\em
  Despite comprising more variables and constraints, the FACeTS
  formulation yields the best performance, achieving a 4x speedup over
  FBSM, and an impressive 40x reduction in the number of
  branch-and-bound nodes.} This represents a key technical
contribution of this paper.

Figures \ref{fig:gap_over_time}, \ref{fig:primal_over_time} and \ref{fig:dual_over_time} further depict the evolution of each formulation's optimality gap, primal bound and dual bound over time, respectively.
The optimality gap after is defined as
$$gap = \frac{Z^{p} - Z^{d}}{|Z^{p}|},$$ where $Z^{p}, Z^{d}$ denote the current primal and dual bounds.
Figures \ref{fig:primal_over_time} and \ref{fig:dual_over_time} show that the computational speedup of FACeTS primarily arises from the enhancements of the dual bound. This improvement is a direct consequence of incorporating the facet-defining inequalities into the formulation, which serve to tighten the problem's relaxation through more precise cuts.

\begin{figure}[!t]
    \centering
    \includegraphics[width=0.80\columnwidth]{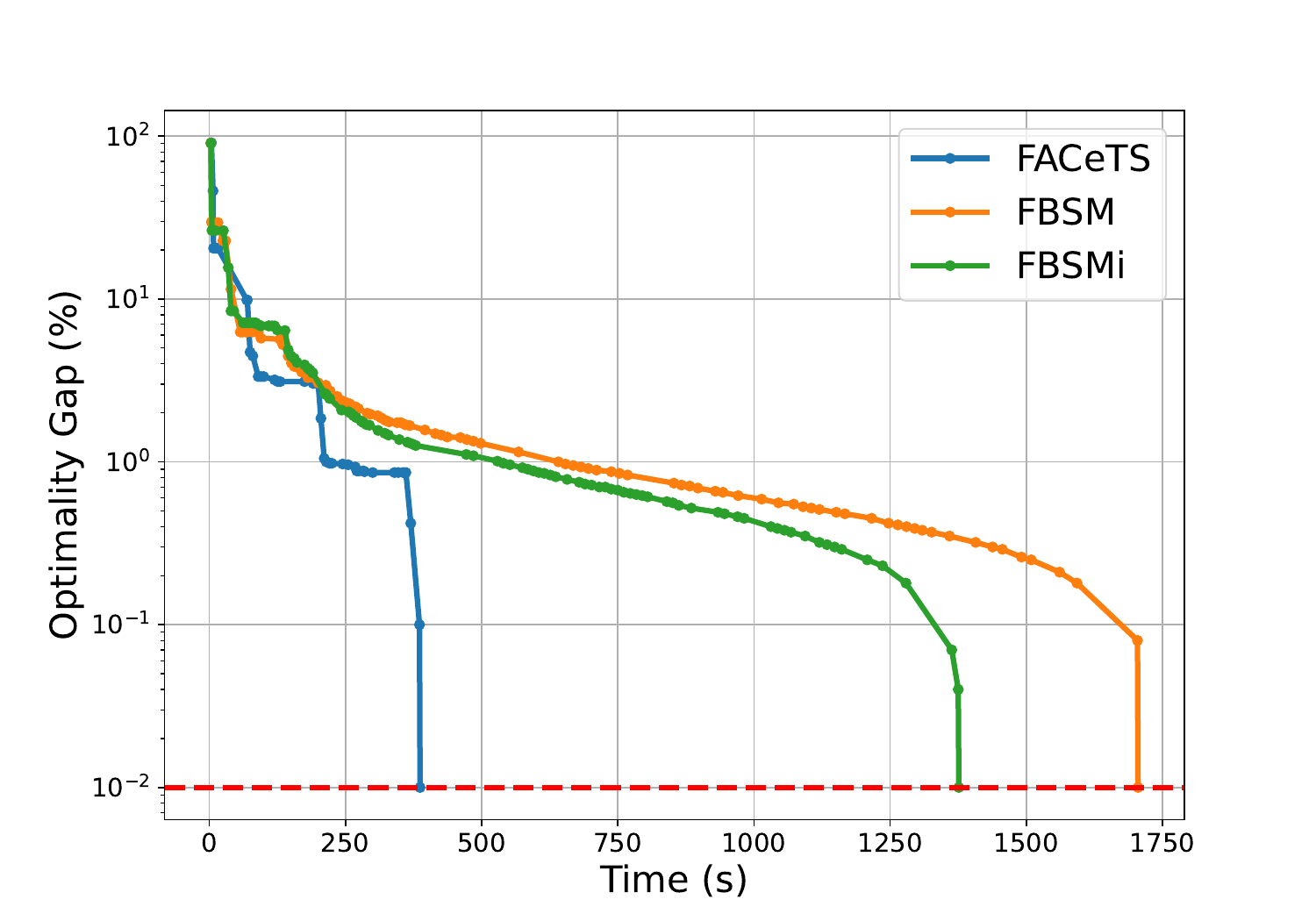}
    \caption{Evolution of optimality gap (in \%) over time, shown in log-scale. The \textcolor{red}{\textbf{red}} dashed line indicates the optimality tolerance of 0.01\%.}
    \label{fig:gap_over_time}
\end{figure}

\begin{figure}[!t]
    \centering
    \includegraphics[width=0.80\columnwidth]{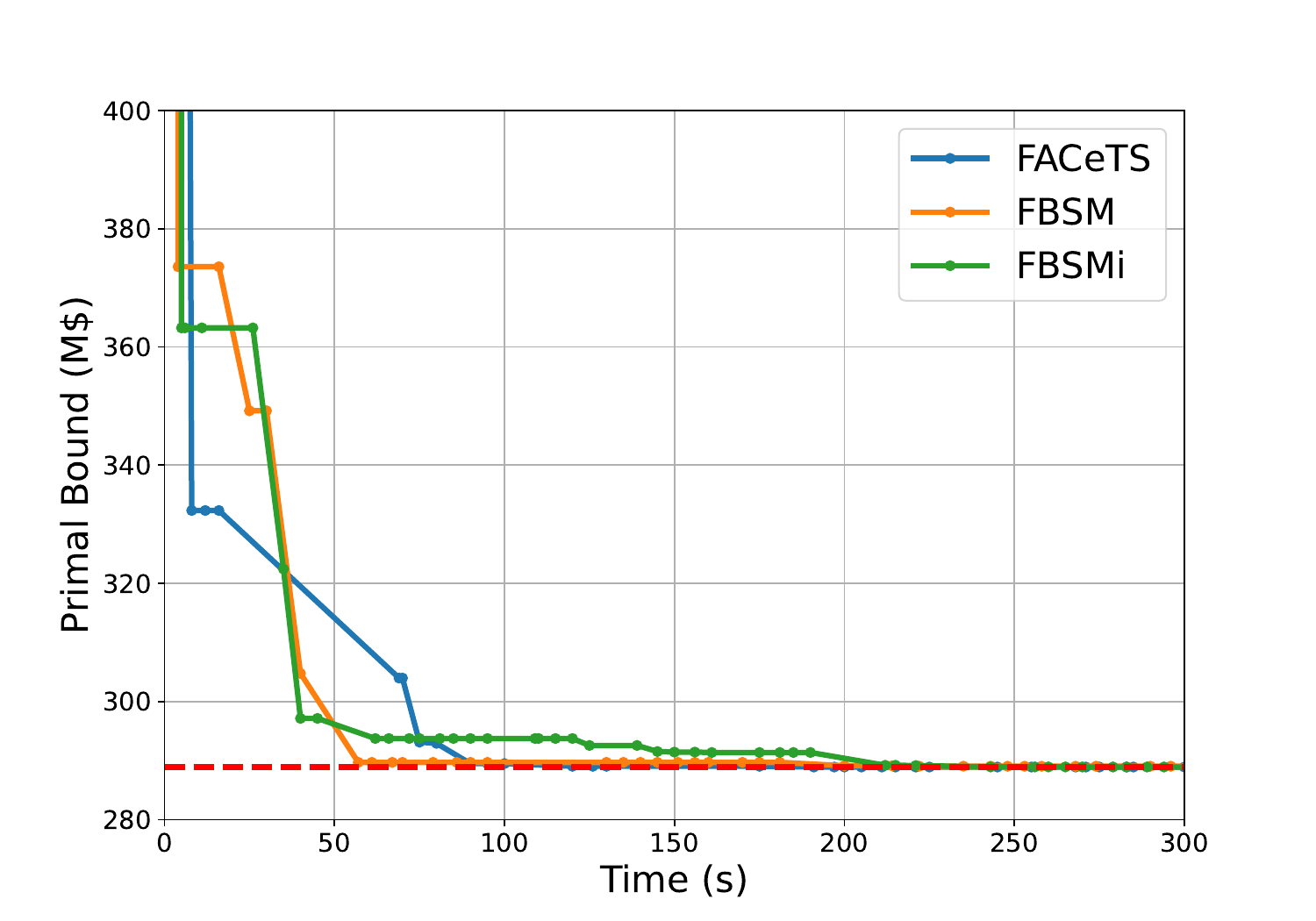}
    \caption{Evolution of each formulation's primal bound (in M\$) over time. The \textcolor{red}{\textbf{red}} dashed line indicates the optimal objective value.}
    \label{fig:primal_over_time}
\end{figure}

\begin{figure}[!t]
    \centering
    \includegraphics[width=0.8\columnwidth]{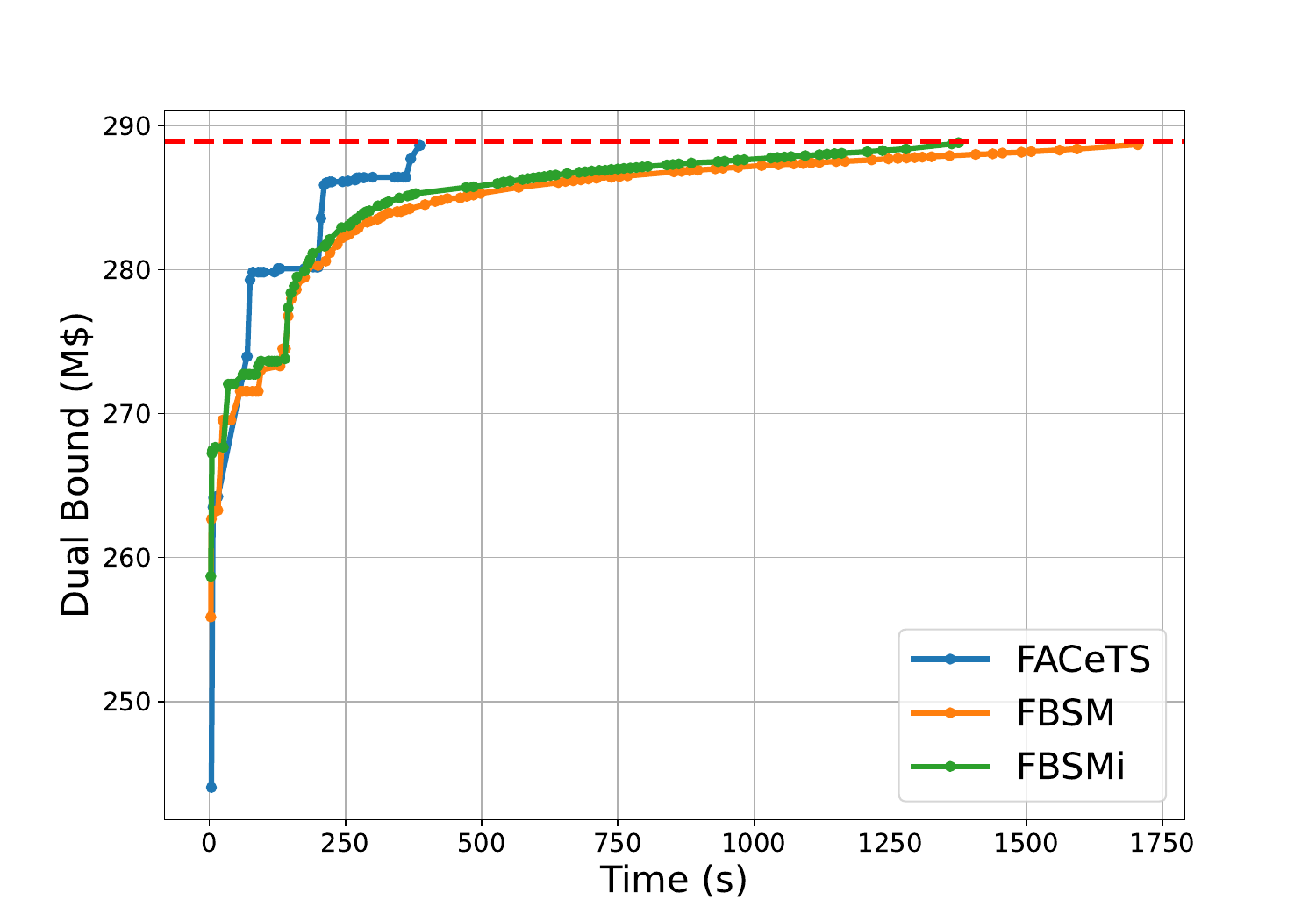}
    \caption{Evolution of each formulation's dual bound (in M\$) over time. The \textcolor{red}{\textbf{red}} dashed line indicates the optimal objective value.}
    \label{fig:dual_over_time}
\end{figure}
    
\subsection{Transmission Expansion Results}
\label{sec:results:expansion_results}

Table \ref{tab:physical_performance} provides a detailed comparison of
the physical performance between TNEP and TNEP+FACTS.  The reported
metrics include the final objective (Total Costs), line capacity
upgrade costs, TCSC installation costs, nonrenewable generation costs,
and the quantity of unserved energy and curtailed renewable energy.
While the decline in total costs associated with the TNEP+FACTS
approach is relatively modest, amounting to approximately 2.5\%, there
are notable advantages in other performance metrics.  Specifically,
the TNEP+FACTS methodology leads to a reduction of 393 MWh in unserved
energy, i.e., a decrease of about 7.4\%.  Moreover, {\em
  curtailed renewable energy decreases by 1057 MWh, i.e., a
  reduction of nearly 12.5\% compared to TNEP.}

    \begin{table}[!t]
        \centering
        \caption{Comparison of Physical Performance}
        \label{tab:physical_performance}
        \begin{tabular}{crr}
        \toprule
                        Metric       & TNEP & TNEP+FACTS \\
        \midrule
        Total Costs (inc. penalties)    & \$296.20M    & \$288.91M            \\ 
        Capacity Upgrade Costs & \$28.23M    & \$27.04M            \\ 
        TCSC Costs             & -           & \$13.58M            \\ 
        Nonrenewable Costs     & \$1.86M     & \$1.84M            \\ 
        \midrule
        Unserved Energy        & 5322 MWh    & 4929 MWh            \\ 
        Curtailed Energy       & 8465 MWh    & 7408 MWh            \\ 
        \bottomrule
        \end{tabular}
    \end{table}

Tables \ref{tab:tnep_branch_investments} and \ref{tab:tnep+facts_branch_investments} provide, for TNEP and TNEP+FACTS, respectively, a detailed breakdown of which branches are slated for capacity upgrades and, for TNEP+FACTS, where TCSCs devices are installed.
Three levels of capacity upgrades (300MW, 600MW and 900MW) are available.
While, at first glance, it might seem
counterintuitive that TNEP+FACTS results in more line capacity
upgrades, Table \ref{tab:physical_performance} reveals that the
overall cost of line upgrades is lower compared to TNEP. Given that
capacity upgrade costs are related to MW-km, this suggests that
TNEP+FACTS capacity upgrades are predominantly occurring on shorter
lines. The geographical distribution of these investments is shown in
Figures \ref{fig:trad_investments} and \ref{fig:facts_investments},
displaying TNEP and TNEP+FACTS investments, respectively. Red lines
indicate the location of capacity upgrades, with the thickness/width
of each line representing one of the three available upgrade
levels. Figure \ref{fig:facts_investments} also shows TCSC
installation locations, denoted by blue 'X's.

\begin{table}[t]
    \centering
    \caption{TNEP Branch Investments}
    \label{tab:tnep_branch_investments}
    \begin{tabular}{lcl}
        \toprule
            Type of Upgrade & Branch\# & From Bus - To Bus\\
        \midrule
            +300MW
                & 151 & San Antonio 14 - San Antonio 2\\ 
                & 166 & San Antonio 2 - San Antonio 33\\
                & 183 & Houston 12 - Houston 45\\
                & 188 & Mexia - Frost\\
                & 194 & Houston 48 - Pasadena 4\\
                & 250 & Abilene 5 - Merkel 2\\
                & 251 & Snyder 2 - Roscoe 2\\
        \midrule
            +600MW
                & 133 & Blum - Frost\\ 
                & 247 & Eastland - Gordon\\
        \midrule
            +900MW
                & 162 & Mexia - Waco 2\\
                & 252 & Throckmorton - Knox City\\
        \midrule
            TCSCs           
                & -- \\
        \bottomrule
    \end{tabular}
\end{table}

\begin{table}[t]
    \centering
    \caption{TNEP+FACTS Branch Investments}
    \label{tab:tnep+facts_branch_investments}
    \begin{tabular}{lcl}
        \toprule
            Type of Upgrade & Branch\# & From Bus - To Bus\\
        \midrule
            +300MW
                & 101 & Abilene 2 - Abilene 5\\
                & 151 & San Antonio 14 - San Antonio 2\\ 
                & 166 & San Antonio 2 - San Antonio 33\\
                & 183 & Houston 12 - Houston 45\\
                & 188 & Mexia - Frost\\
                & 194 & Houston 48 - Pasadena 4\\
                & 251 & Snyder 2 - Roscoe 2\\
                & 252 & Throckmorton - Knox City\\
        \midrule
            +600MW
                & 133 & Blum - Frost\\ 
                & 250 & Eastland - Gordon\\
        \midrule
            +900MW
                & 162 & Mexia - Waco 2\\
                & 247 & Eastland - Gordon\\
        \midrule
            TCSCs           
                & \phantom{0}72 & Mccamey - Carlsbad\\
                & 197 & Throckmorton - Knox City\\
                & 198 & Throckmorton - Knox City\\
                & 235 & North Richland Hills 1 - Weatherford 4\\
                & 248 & Mccamey - Carlsbad\\
        \bottomrule
    \end{tabular}
\end{table}

\begin{figure}[!t]
    \centering
    \includegraphics[width=0.80\columnwidth]{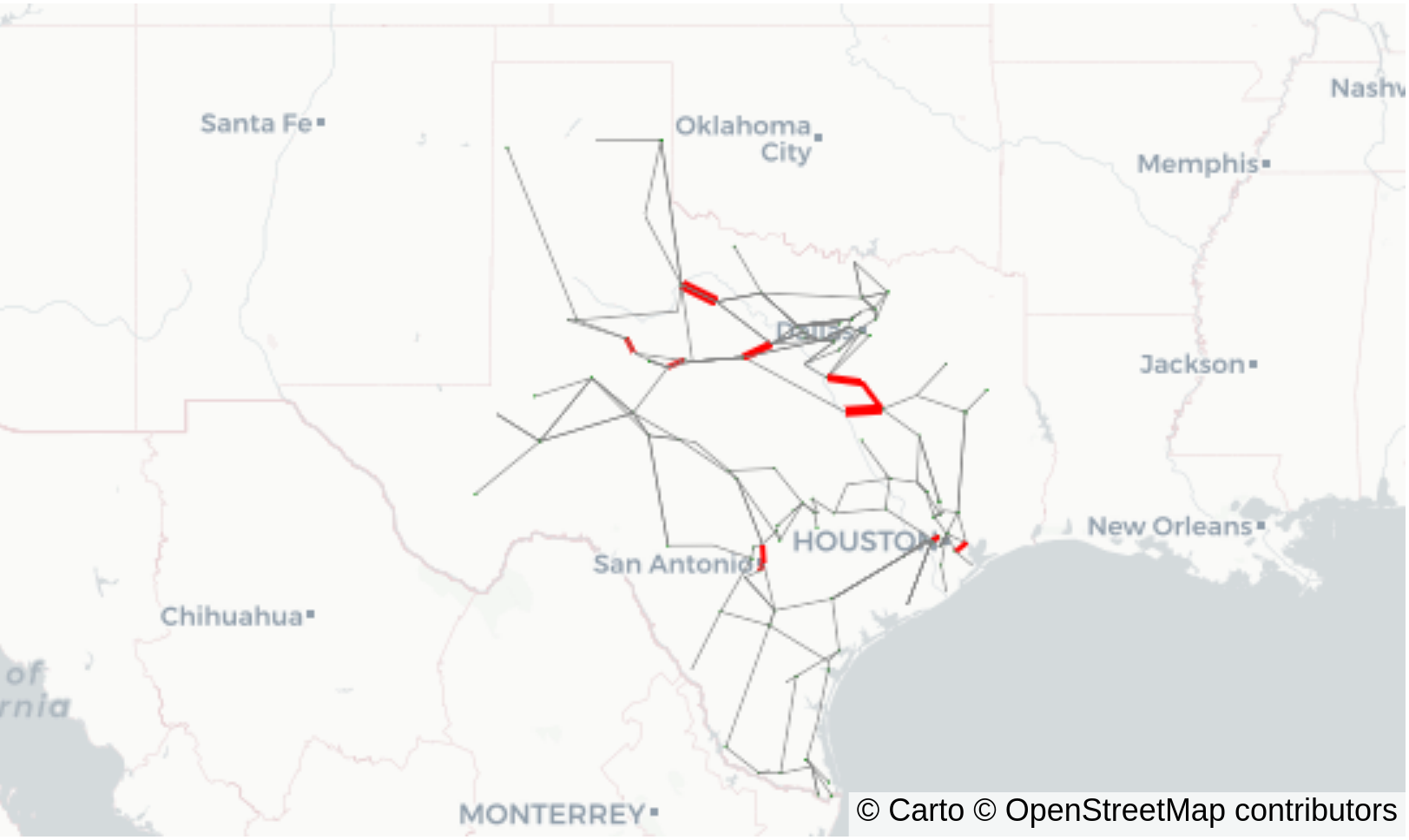}
    \caption{Baseline TNEP Investments. \textcolor{red}{\textbf{Red}} lines indicate the location of capacity upgrades, with the thickness/width of each line representing one of the three available upgrade levels.}
    \label{fig:trad_investments}
\end{figure}

\begin{figure}[!t]
    \centering
    \includegraphics[width=0.80\columnwidth]{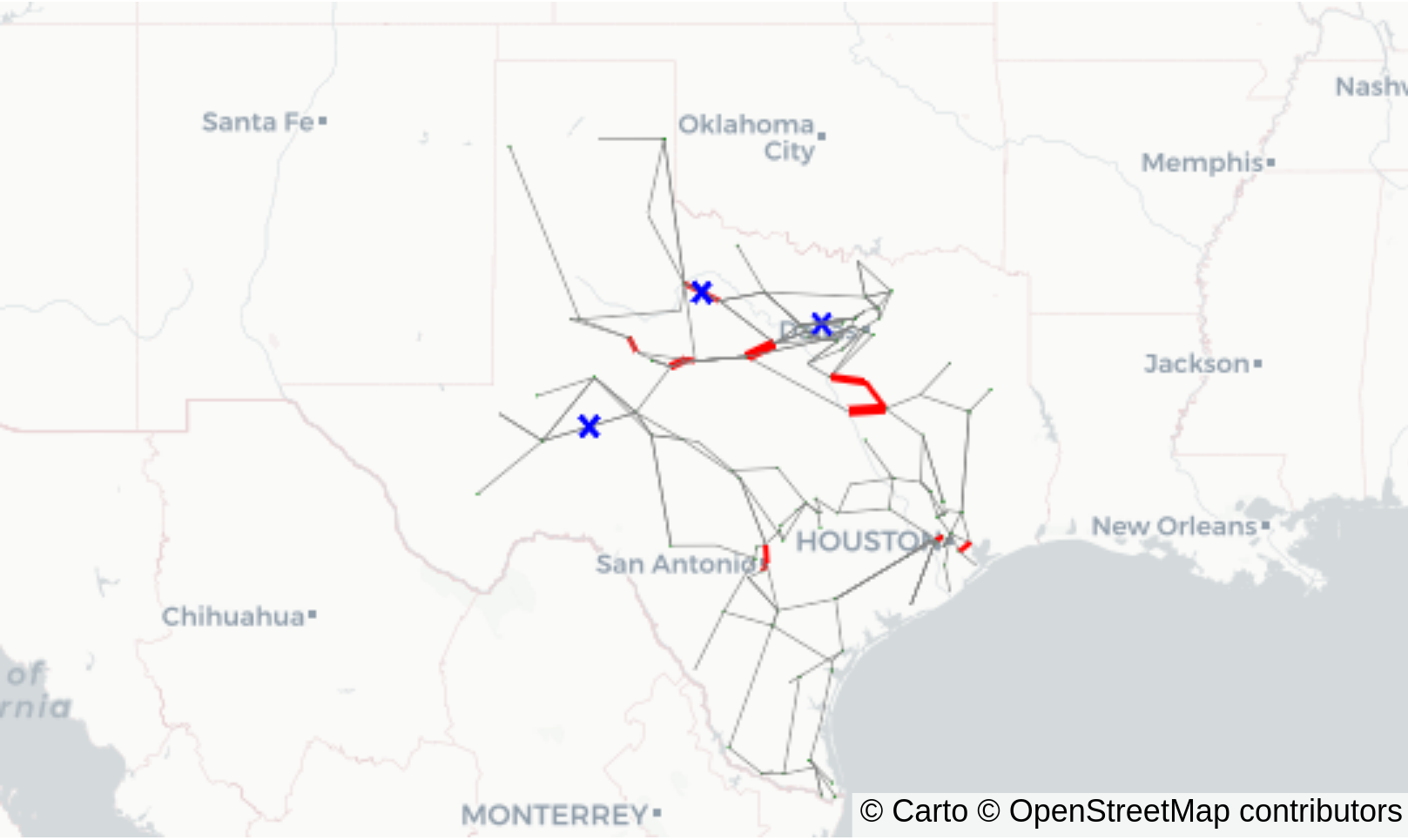}
    \caption{TNEP+FACTS Investments. \textcolor{red}{\textbf{Red}} lines indicate the location of capacity upgrades, with the thickness/width of each line representing one of the three available upgrade levels. \textcolor{blue}{\textbf{Blue}} 'X's indicate the location of installed TCSCs.}
    \label{fig:facts_investments}
\end{figure}

    Figures \ref{fig:trad_ue} and \ref{fig:facts_ue} depict the location and amount of mean unserved energy for TNEP and TNEP+FACTS, respectively.  A predominant portion of the unserved energy is observed near urban load centers. Notably, the integration of FACTS results in a pronounced reduction of unserved energy in the northern corridors. This aligns with the locations where the TCSCs were installed as seen in Figure \ref{fig:facts_investments}.

\begin{figure}[t]
    \centering
    \includegraphics[width=0.80\columnwidth]{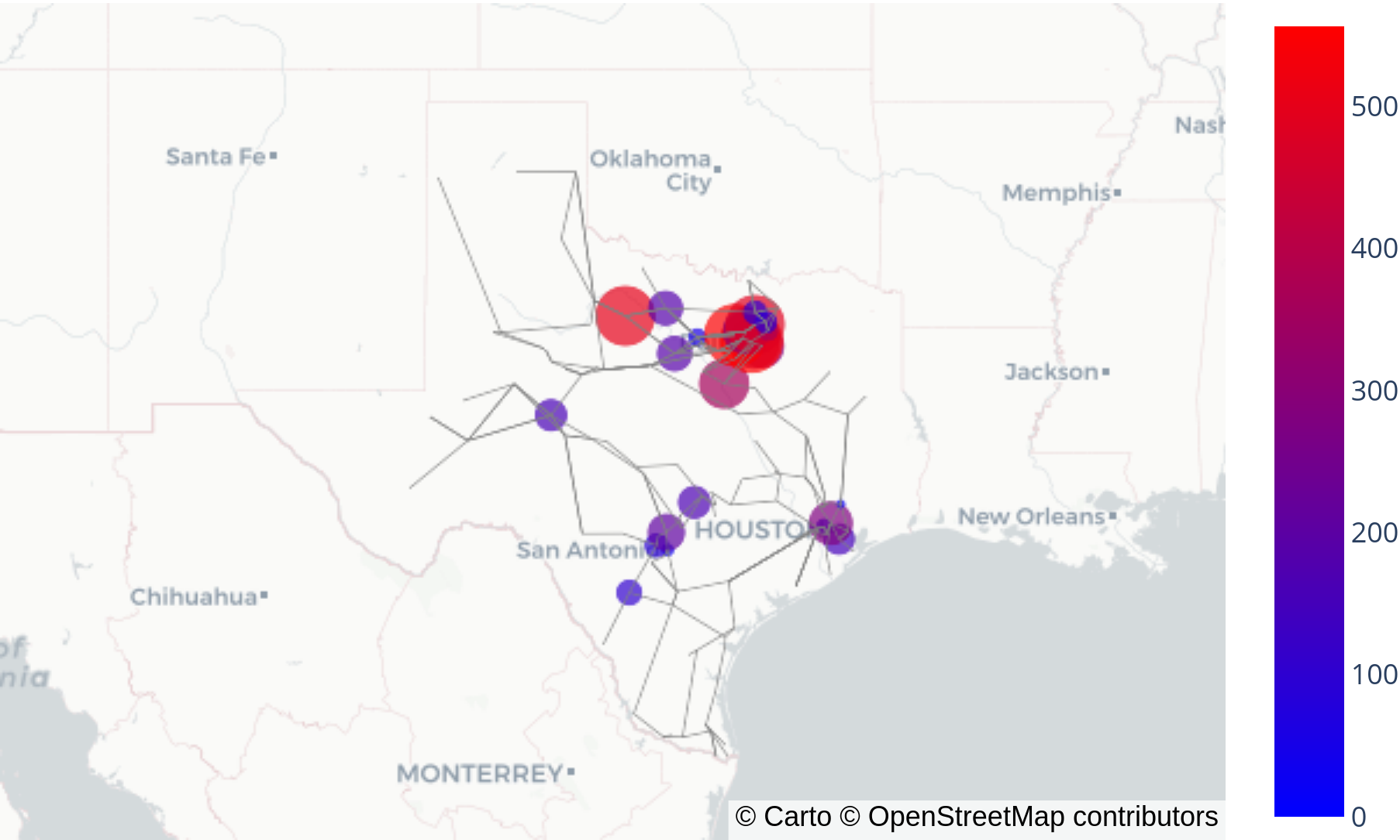}
    \caption{Unserved energy in baseline TNEP. Circles indicate the location and amount (in MWh) of average unserved energy.}
    \label{fig:trad_ue}
\end{figure}

\begin{figure}[t]
    \centering
    \includegraphics[width=0.80\columnwidth]{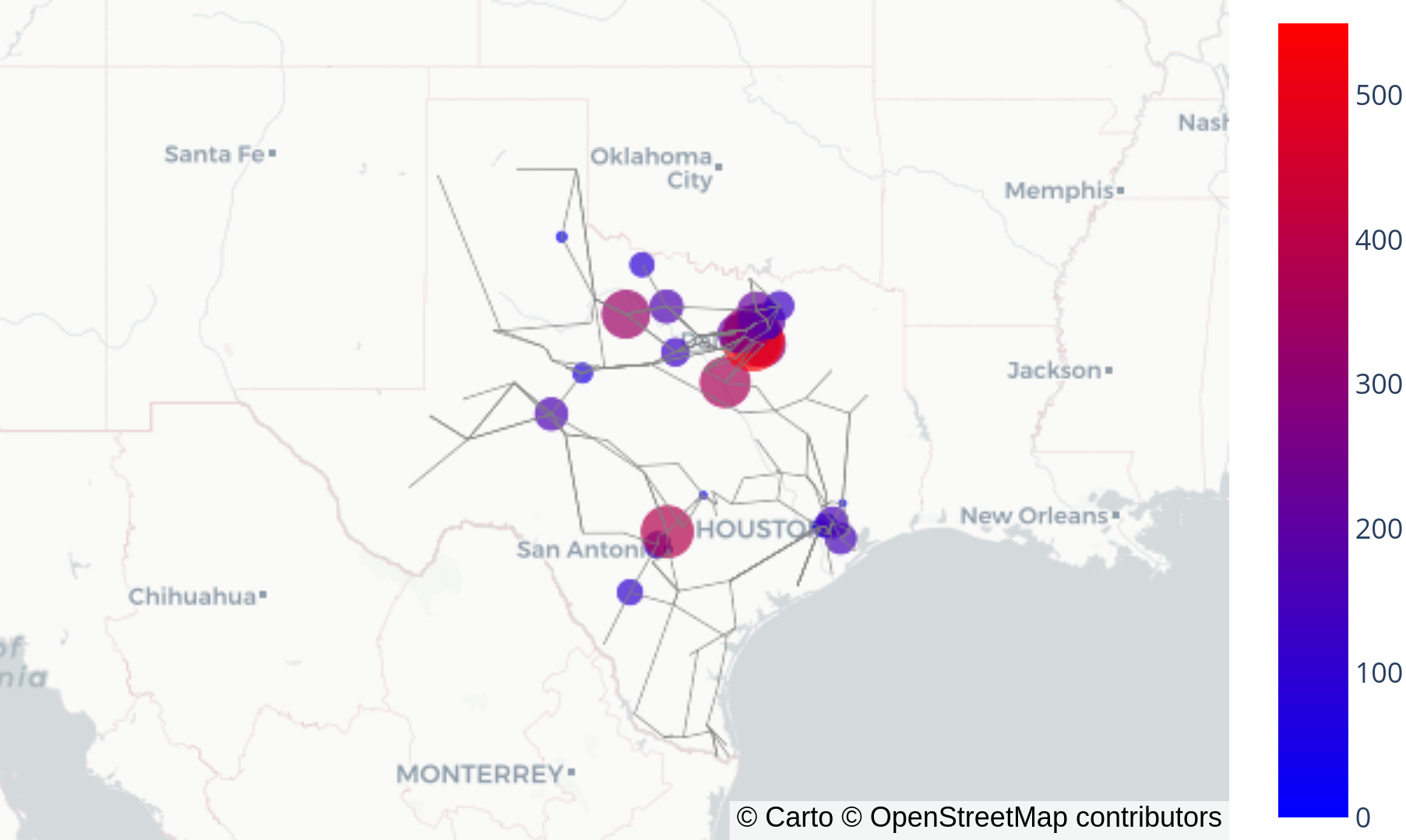}
    \caption{Unserved energy in TNEP+FACTS. Circles indicate the location and amount (in MWh) of average unserved energy.}
    \label{fig:facts_ue}
\end{figure}

    Similarly, Figures \ref{fig:trad_curtail} and \ref{fig:facts_curtail} provide insight into the geographic distribution of mean curtailed energy. Notably, a significant portion of this curtailed energy is concentrated in the western part of Texas, aligning with the areas of renewable capacity depicted in Figure \ref{fig:texas_renewables}. The inclusion of FACTS further leads to a tangible reduction in this curtailed energy.

\begin{figure}[t]
    \centering
    \includegraphics[width=0.80\columnwidth]{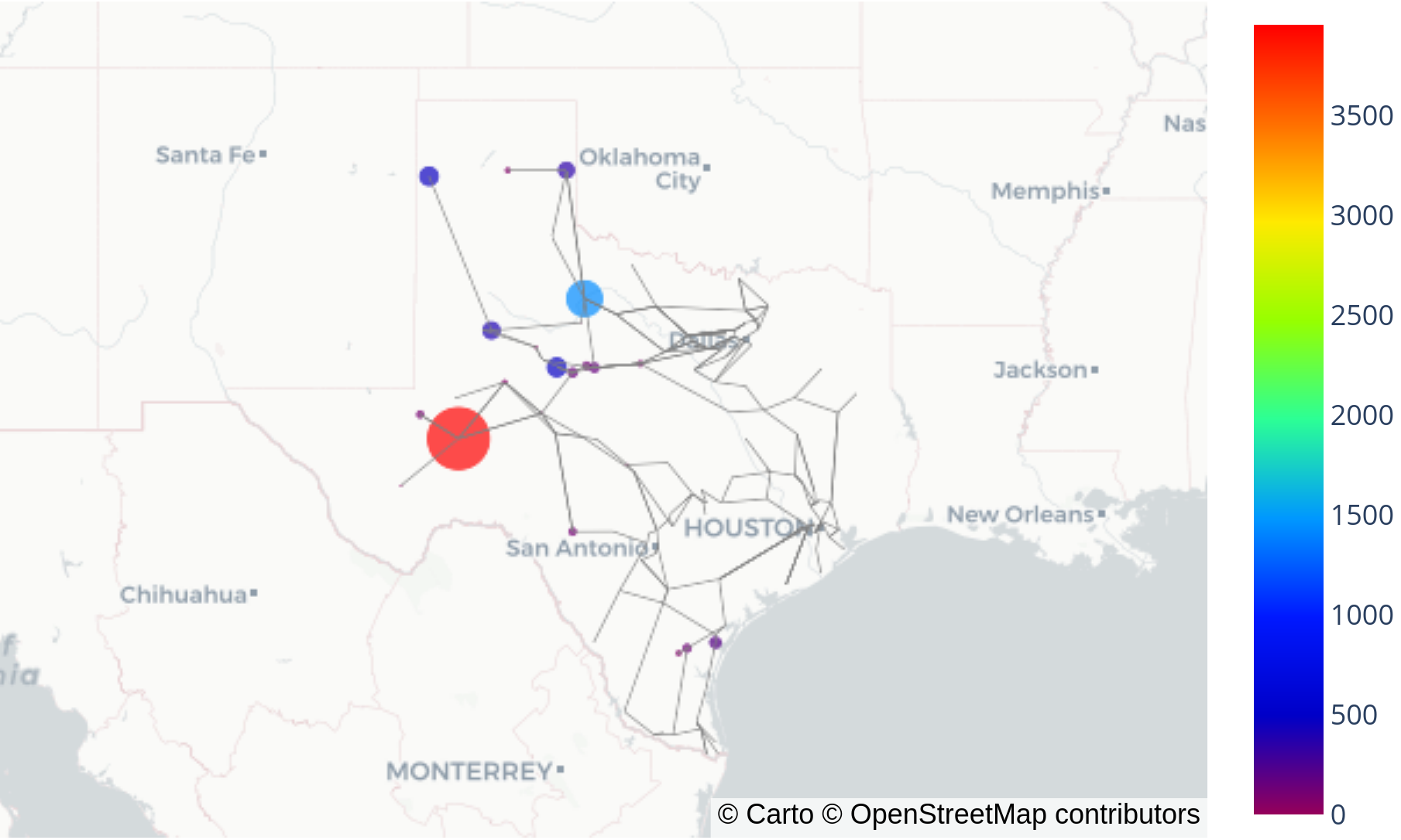}
    \caption{Curtailed energy in baseline TNEP. Circles indicate the location and amount (in MWh) of average curtailed energy.}
    \label{fig:trad_curtail}
\end{figure}

\begin{figure}[t]
    \centering
    \includegraphics[width=0.80\columnwidth]{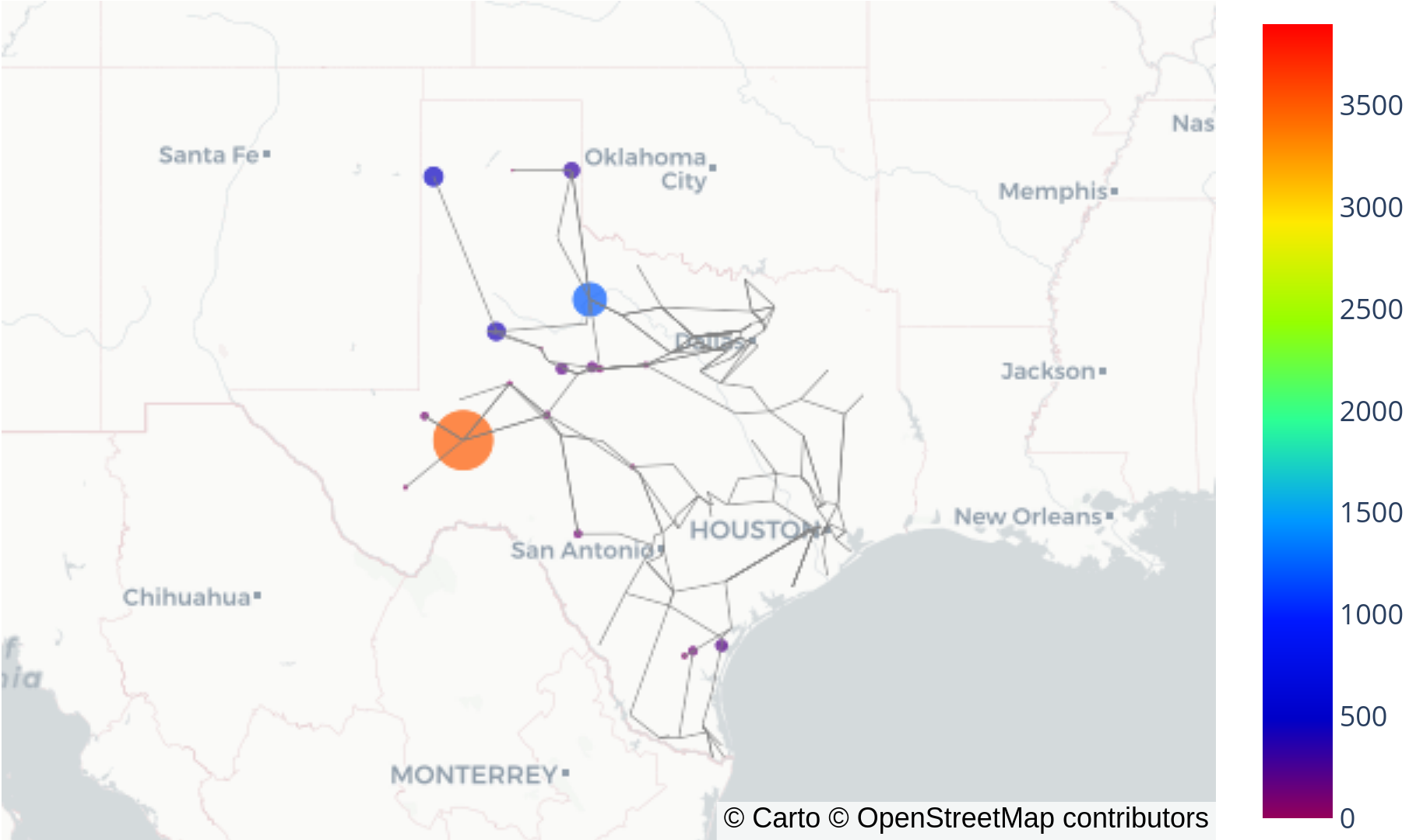}
    \caption{Curtailed energy in baseline TNEP. Circles indicate the location and amount (in MWh) of average curtailed energy.}
    \label{fig:facts_curtail}
\end{figure}

{\em In summary, by collectively analyzing the figures, it becomes
  evident that integrating FACTS to TNEP facilitates the transfer of
  renewable capacity to Texas's major urban centers.}

\section{Conclusion}
\label{sec:conclusion}

The paper has presented a novel MILP formulation for including FACTS
devices in TNEP problems, which directly represents the change in
power flow induced by FACTS devices.  The proposed formulation enjoys
strong theoretical guarantees, namely, it uses facet-defining
constraints instead of weak big-M constraints as in previous
formulations.  Numerical experiments demonstrate the superiority of
the proposed approach, which yields a 4x speedup and a 40x reduction
in branch-and-bound nodes compared to state-of-the-art formulations.
In addition, the results show that FACTS devices have the potential to
reduce load shedding and renewable generation curtailment, by reducing
congestion on the system.

Future work will explore the integration of energy storage technologies within the TNEP+FACTS framework, and conduct experiments across more detailed and larger-scale systems, for which further computational enhancements will likely be needed.
Another interesting direction is to extend the proposed formulation to consider FACTS devices in unit commitment and economic dispatch problems, which underlie most day-ahead and real-time electricity markets in the US.
Tractable formulations for such problems would allow for more efficient markets, and expand the set of control actions available to operators in real-time, showcasing the potential for improved reliability and efficiency of power systems.

\bibliographystyle{IEEEtran}
\bibliography{refs}

\begin{thebibliography}{10}
\providecommand{\url}[1]{#1}
\csname url@samestyle\endcsname
\providecommand{\newblock}{\relax}
\providecommand{\bibinfo}[2]{#2}
\providecommand{\BIBentrySTDinterwordspacing}{\spaceskip=0pt\relax}
\providecommand{\BIBentryALTinterwordstretchfactor}{4}
\providecommand{\BIBentryALTinterwordspacing}{\spaceskip=\fontdimen2\font plus
\BIBentryALTinterwordstretchfactor\fontdimen3\font minus \fontdimen4\font\relax}
\providecommand{\BIBforeignlanguage}[2]{{%
\expandafter\ifx\csname l@#1\endcsname\relax
\typeout{** WARNING: IEEEtran.bst: No hyphenation pattern has been}%
\typeout{** loaded for the language `#1'. Using the pattern for}%
\typeout{** the default language instead.}%
\else
\language=\csname l@#1\endcsname
\fi
#2}}
\providecommand{\BIBdecl}{\relax}
\BIBdecl

\bibitem{RTE}
{Réseau de Transport d'Électricité}, ``{Energy Pathways to 2050},'' {Réseau de Transport d'Électricité}, Tech. Rep., 2022.

\bibitem{Gideon2019_SurveyTNEP}
N.~Gideon~Ude, H.~Yskandar, and R.~Coneth~Graham, ``A comprehensive state-of-the-art survey on the transmission network expansion planning optimization algorithms,'' \emph{IEEE Access}, vol.~7, pp. 123\,158--123\,181, 2019.

\bibitem{Mokhtari}
M.~S. Mokhtari, M.~Gitizadeh, and M.~Lehtonen, ``Optimal coordination of thyristor controlled series compensation and transmission expansion planning: Distributionally robust optimization approach,'' \emph{Electric Power Systems Research}, vol. 196, p. 107189, 2021.

\bibitem{Li}
C.~Li, A.~J. Conejo, P.~Liu, B.~P. Omell, J.~D. Siirola, and I.~E. Grossmann, ``Mixed-integer linear programming models and algorithms for generation and transmission expansion planning of power systems,'' \emph{European Journal of Operational Research}, vol. 297, no.~3, pp. 1071--1082, 2022.

\bibitem{Micheli}
G.~Micheli, M.~T. Vespucci, M.~Stabile, C.~Puglisi, and A.~Ramos, ``A two-stage stochastic milp model for generation and transmission expansion planning with high shares of renewables,'' \emph{Energy Systems}, pp. 1--43, 2020.

\bibitem{Lumbreras}
S.~Lumbreras and A.~Ramos, ``Transmission expansion planning using an efficient version of benders' decomposition. a case study,'' in \emph{2013 IEEE Grenoble Conference}.\hskip 1em plus 0.5em minus 0.4em\relax IEEE, 2013, pp. 1--7.

\bibitem{Blanchot}
X.~Blanchot, F.~Clautiaux, A.~Froger, and M.~Ruiz, ``Solving a bilevel stochastic generation and transmission expansion planning problem,'' 2023, preprint available on HAL-Inria: \url{https://hal-lirmm.ccsd.cnrs.fr/INRIA/hal-03957750v1}.

\bibitem{de_Araujo}
R.~A. de~Araujo, S.~P. Torres, J.~Pissolato~Filho, C.~A. Castro, and D.~Van~Hertem, ``Unified ac transmission expansion planning formulation incorporating vsc-mtdc, facts devices, and reactive power compensation,'' \emph{Electric Power Systems Research}, vol. 216, p. 109017, 2023.

\bibitem{Luburic}
Z.~Luburi{\'c}, H.~Pand{\v{z}}i{\'c}, and M.~Carri{\'o}n, ``{Transmission expansion planning model considering battery energy storage, TCSC and lines using AC OPF},'' \emph{IEEE Access}, vol.~8, pp. 203\,429--203\,439, 2020.

\bibitem{Esmaili}
M.~Esmaili, M.~Ghamsari-Yazdel, N.~Amjady, C.~Chung, and A.~J. Conejo, ``Transmission expansion planning including tcscs and sfcls: A minlp approach,'' \emph{IEEE Transactions on Power Systems}, vol.~35, no.~6, pp. 4396--4407, 2020.

\bibitem{Ziaee}
O.~Ziaee, O.~Alizadeh-Mousavi, and F.~F. Choobineh, ``Co-optimization of transmission expansion planning and tcsc placement considering the correlation between wind and demand scenarios,'' \emph{IEEE Transactions on Power Systems}, vol.~33, no.~1, pp. 206--215, 2017.

\bibitem{Franken}
M.~Franken, H.~Barrios, A.~B. Schrief, and A.~Moser, ``Transmission expansion planning via power flow controlling technologies,'' \emph{IET Generation, Transmission \& Distribution}, vol.~14, no.~17, pp. 3530--3538, 2020.

\bibitem{TX123BT}
\BIBentryALTinterwordspacing
J.~Lu, ``Tx-123bt network model,'' 2023. [Online]. Available: \url{https://rpglab.github.io/resources/TX-123BT/}
\BIBentrySTDinterwordspacing

\bibitem{Lubin2023}
M.~Lubin, O.~Dowson, J.~{Dias Garcia}, J.~Huchette, B.~Legat, and J.~P. Vielma, ``{JuMP} 1.0: {R}ecent improvements to a modeling language for mathematical optimization,'' \emph{Mathematical Programming Computation}, 2023.

\bibitem{PACE}
\BIBentryALTinterwordspacing
PACE, \emph{{P}artnership for an {A}dvanced {C}omputing {E}nvironment ({PACE})}, 2017. [Online]. Available: \url{http://www.pace.gatech.edu}
\BIBentrySTDinterwordspacing

\bibitem{gurobi}
\BIBentryALTinterwordspacing
{Gurobi Optimization, LLC}, ``{Gurobi Optimizer Reference Manual},'' 2023. [Online]. Available: \url{https://www.gurobi.com}
\BIBentrySTDinterwordspacing

\bibitem{highs}
Q.~Huangfu and J.~J. Hall, ``Parallelizing the dual revised simplex method,'' \emph{Mathematical Programming Computation}, vol.~10, no.~1, pp. 119--142, 2018.

\end{thebibliography}

\end{document}